\RequirePackage{fix-cm}
\documentclass{svjour3}                     
\smartqed  
\usepackage{graphicx}
\usepackage{latexsym}
\usepackage{amsmath,amssymb}
 \journalname{Cryptography and Communications}
\begin{document}

\title{Lifted Codes and Lattices from Codes Over Finite Chain Rings
}


\author{Reguia Lamia Bouzara        \and
        Kenza Guenda \and Edgar Mart\'inez-Moro
}


\institute{R.L. Bouzara \and K. Guenda\at
              Faculty of Mathematics USTHB\\ University of Science and Technology of Algiers,
              Algeria\\
              \email{bouzaralamia@outlook.fr, kguenda@usthb.dz}           
           \and
           E. Mart\'inez-Moro\at
             Institute of Mathematics, University of Valladolid, Spain. \\
             \email{edgar.martinez@uva.es}\\
             Partially supported by the Spanish State Research Agency (AEI) under
             Grant PGC2018-096446-B-C21.
}

\date{Received: \today / Accepted: date}

\maketitle

\begin{abstract}
In this paper we give the generalization of lifted codes over any finite chain ring. This has been done by using the construction of finite chain rings from $p$-adic fields. Further we propose a lattice construction from linear codes over finite chain rings using lifted codes. 
\end{abstract}

\section{Introduction} \label{intro}

Codes over finite rings have received significant attention in recent decades. Several authors have studied these codes due to their relationship with lattices construction, among other properties. The class of $p$-adic codes was introduced in \cite{clod}. Calderbank and Sloane investigated codes over $p$-adic integers and studied lifts of codes over $\mathbb{F}_p$ and $\mathbb{Z}_{p^e}$.

 Lifted codes over finite chain rings were studied in \cite{douglif}, however this study was restricted to the finite chain rings of the form $\mathbb{F}_q[t]/\langle t^k\rangle$ as pointed out by the reviewer in \cite{Wood}. Later, Young Ho generalized the construction of cyclic lifted codes for arbitrary finite fields to codes over Galois rings $GR(p^e,r)$ in \cite{cyclif}.
 
 In \cite{Wood}, the reviewer stated that a unified treatment valid for all chain rings would certainly be desirable. Therefore, this study investigates the structure of finite chain rings as non-trivial quotient of ring integers of $p$-adic fields to generalize the construction in \cite{douglif}.
 
 A finite commutative chain ring is a finite local ring whose maximal ideals are principal. Any finite chain ring can be constructed from $p$-adic fields (see for example \cite{hou1}) as follows: let $K$ be a finite extension of the field of $p$-adic numbers $\mathbb{Q}_p$ with residue degree $r$ and ramification index $s$, let $\mathcal{O}_K$ be the ring of integers of $K$ and let $\pi$ be a prime of $K$. Then $\mathcal{O}_K/\langle\pi^{(n-1)s+t}\rangle$ is a finite commutative chain ring of invariants  $(p,n,r,s,t)$. Every finite commutative chain ring can be obtained in this way.
 
 This paper defines finite chain rings as non-trivial quotient of ring integers of $p$-adic fields to provide a general and a unified treatment of lifted codes for any finite chain rin. This definiton also allows to introduce a general construction of lifted cyclic codes that can be used to lift codes over finite fields $\mathbb{F}_p^r$ to codes over finite chain rings. Then the lifted codes are used to give a general construction $A$ of lattices from codes over finite chain rings that generalizes the construction of lattices in \cite{oggier}.\\
 The structure of the paper is as follows: Section~\ref{sec:1} introduces the $p$-adic fields and their extensions and, Section~\ref{sec:codes} describes the connection between $p$-adic fields and the construction of a finite chain ring, based on this fact, the construction of lifted codes is generalized. Section~\ref{sec:3} provides some definitions of lattices over $p$-adic integers showing a lattice construction over finite chain rings from codes and its properties and finally, Section~\ref{sec:5} concludes the paper.

\section{Preliminaries} \label{sec:1}

\subsection{$p$-adic fields}
The results presented in this section can be found in \cite{iwasawa,book}. 
Let $p$ be a prime number, and let $x$ be an element of the rational field $\mathbb{Q}$, then, $x$ can be written in a unique way as $x=p^{a}\frac{m'}{n'},$ where, $m' \in \mathbb{Z}^{*}$ and $n'\in \mathbb{N}^{*}$, and  both $m'$ and $n'$ are not divisible by $p$.
The \emph{ $p$-adic valuation} $v_p(x)$ of $x\in \mathbb Q$ is defined as: 
\begin{equation*}
    v_p(x)=
    \begin{cases}
      +\infty & \hbox{ if } x=0 \\
      a     & \hbox{ if }  x=p^{a}\frac{m'}{n'}, p\nmid m' \cdot n' 
    \end{cases}
  \end{equation*}
and the $p$-adic norm (which is an ultra-metric absolute value) is given by: 
$\vert x \vert_p=p^{-v_{p}(x)}$.
The completion of $\mathbb{Q}$ by the absolute value $\vert \cdot\vert_p$ is the field of $p$-adic numbers and is denoted as usual by $\mathbb{Q}_p$. The next proposition characterize the field $\mathbb{Q}_p$ as follows: 

\begin{proposition} Given $v_p$, $\vert \cdot\vert_p$ and $\mathbb{Q}_p$ defined as above:
\begin{enumerate}
\item The set $$\mathbb{Z}_p=\lbrace a\in \mathbb{Q}_p\mid v(a)\geqslant 0\rbrace=\lbrace a\in \mathbb{Q}_p\mid \vert a\vert\leqslant 1\rbrace$$ is a unitary ring  called the valuation ring or the ring of integers of $\mathbb{Q}_p$.
\item $\mathbb{Q}_p=\mathcal{F}(\mathbb{Z}_p)$ where $\mathcal{F}(\mathbb{Z}_p)$ is the fraction ring of $\mathbb{Z}_p.$

\item The set $\mathfrak{m}= \lbrace a\in \mathbb{Q}_p\mid v(a)> 0\rbrace=\lbrace a\in \mathbb{Q}_p\mid\vert a\vert < 1\rbrace,$ is the unique maximal ideal of $\mathbb{Z}_p$ and it is generated by $p$, i.e. $\mathfrak{m}=\langle p\rangle$. The ideal 
$\mathfrak{m}$ is called ideal of the valuation $v_p$.

\item By the previous item,   $\mathbb{Z}_p$ is a local ring and  the quotient ring $\mathbb{Z}_p/\mathfrak{m}=\mathbb{F}_p$  is the finite field with $p$ elements which is  the residual field of the valuation $v_p$.

\item Any  $a\in \mathbb{Z}_p$  can be written as follows:
$$a=\sum_{i=0}^{+\infty} a_i p^i\textit{ ; }a_i\in \mathbb{F}_p.$$ 
\end{enumerate}
\end{proposition}

Let $K$ be an algebraic extension of $\mathbb{Q}_p$ of degree $n$ and  $x\in K$. $\mathrm{N}_{{\mathbb{Q}_p}|K}(x)$ is
the endomorphism of the $\mathbb{Q}_p$-vector space $K$ defined by multiplication by $x$. The characteristic polynomial of this endomorphism expressed as follows:
$$X^n+...+(-1)^n \mathrm{N}_{{\mathbb{Q}_p}|K}(x),$$
which is cancelled for $x$ in $K$. 
In other words,  $x$ is an integer of $K$ and 
$\mathbb{Z}_p=\mathcal{O}_K\cap \mathbb{Q}_p$, where, $\mathcal{O}_K$ denotes the ring of integers of $K$. Then, the relation
$$w_p(x)=\frac{1}{n}v_p(\mathrm{N}_{{\mathbb{Q}_p}|K}(x))$$
defines the unique valuation $w_p$ of $K$ extending the valuation $v_p$ of $\mathbb{Q}_p$.
Let $u$ be the restriction of the valuation $w_p$ to the elements in  $\mathbb{Q}_p$. The \emph{ramification index} of  $K|\mathbb{Q}_p$ is defined as follows:
$$e=[w(K^{*}):u(\mathbb{Q}_p^{*})]=w_p(p).$$

\begin{proposition} Let $K$ be an algebraic extension of $\mathbb{Q}_p$ of degree $n$. Then:
\begin{enumerate} 
\item The ring of integers of $K$ is given by
$$\mathcal{O}_K=\lbrace a\in \mathbb{Q}_p\mid  v_p(a)\geqslant 0\rbrace.$$
\item The maximal ideal of the ring $\mathcal{O}_K$ is generated by an element $\pi$ called the uniformizer and $\mathfrak{m}_K=\langle\pi\rangle.$
\item The residual field is $\mathcal{O}_K/\mathfrak{m}_K=\mathbb{F}_{p^f}$ and $f$ is considered the inertial degree.
\item If $p$ is an uniformizer of $\mathbb{Q}_p$ and $\pi$ is an uniformizer of $K$, then
 $$\vert p\vert_p=\vert \pi\vert_p^e,$$
where, $e$ is ramification index.
\item $[K:\mathbb{Q}_p]=n=ef$.
\item An element $\alpha$  of $\mathcal{O}_K$ can be written as follows:
$$\alpha= \sum_{i=0}^{+\infty} a_i \pi^{i}$$
where, the $a_i$'s are elements of $\mathcal{O}_K/\mathfrak{m}_K=\mathbb{F}_{p^f}$.
\end{enumerate}
\end{proposition}
Finally, the following theorem is a local version of the fact that if $K$ is a number field, then, $\mathcal{O}_K$ is a free
 $\mathbb{Z}_p$-module of rank $[K:\mathbb{Q}_p].$ 
\begin{theorem}
$\mathbb{Z}_p$-module $\mathcal{O}_K$ is free of rank 
$$n=[K:\mathbb{Q}_p]=ef.$$
If $\lbrace\alpha_1,\ldots ,\alpha_f\rbrace\subset \mathbb{Z}_p$ is a set such that the reductions $\overline{\alpha}_i$ generate $\mathbb{F}_{p^f}$ as an $\mathbb{F}_p-$vector space, the set 
 $\lbrace \alpha_j \pi_K^k\rbrace$ where $0\leq k\leq e$ and $ 1\leq j \leq f$  is an $\mathbb{Z}_p$-basis of  $\mathcal{O}_K.$ 
\end{theorem} 

\subsection{Chain rings and $p$-adic fields}\label{sec:chain}

A finite commutative chain ring is a finite commutative  ring whose ideals form a chain under inclusion. A complete description of this type of rings can be found in \cite{mc}.  Let $R$ be a finite chain ring of characteristic $p^n$ with maximal ideal $\mathfrak{m}$ and nilpotency index $s$. Its residue field is $R/\mathfrak{m}=\mathbb F_{p^r}$ (the finite field with  $p^r$ elements). Every finite chain ring can be written as follows:
$$R=GR(p^n,r)[X]/\langle g(X),p^{n-1}X^t\rangle,$$
where, $GR(p^n,r)$ is the Galois ring of size ${p^n}^r$ and characteristic $p^n$, and  $g(X)$ is an Eisenstein polynomial in  $GR(p^n,r)[X]$ of degree $e$. The integers $(p,n,r,e,t)$ are called the invariants of the finite chain ring $R$. The ring $R$ can also be constructed from $p$-adic fields since finite commutative chain rings are the non-trivial quotients of rings of integers of $p$-adic fields \cite{hou1,hou2}. Next proposition summarizes the connection between $p$-adic fields and finite chain rings.

\begin{proposition}\cite{hou2}
Let K be a finite extension of $\mathbb{Q}_p$ such as $[K:\mathbb{Q}_p]=n$ with residue degree $r$ and ramification index $e$. Let $\mathcal{O}_K$ be the ring of integers of $K$ and $\pi$ a prime of $K$. Then:
$$\mathcal{O}_K/\pi^s\mathcal{O}_K\cong GR(p^n,r)[X]/\langle g(X),p^{n-1}X^t\rangle.$$
\end{proposition}

\section{Finite Chain Rings and Lifted Codes}\label{sec:codes}
A linear code $\mathcal C\subseteq R^m$ of length $m$ over a finite chain ring $R$ is a submodule of $R^m$. The length $m$ is assumed not to be divisible by the characteristic of the residue field $R/\mathfrak{m}=\mathbb{F}_{p^r}$. A matrix $G$ with entries in $R$, is called a generator matrix for the code $\mathcal C$ if its rows span $C$ and none of them can be written as an $R-$linear combination of other remaining rows of $G$. The generator  matrix in standard form if it is written as follows (see \cite{Norton1}):
	$$\begin{pmatrix}
	I_{k_0}  & A_{0,1}    & A_{0,2}       & A_{0,3}       & \cdots    &  \cdots                    &  A_{0,s}     \\
	0        & \pi I_{K_1}& \pi A_{1,2}   & \pi A_{1,3}   & \cdots    &  \cdots                    & \pi A_{1,s}   \\
	0        & 0          & \pi^2I_{k_2}  & \pi^2 A_{2,3} & \cdots    &  \cdots                    & \pi^2 A_{2,s}  \\
	\vdots   & \vdots     & 0             & \ddots        & \ddots    &                            & \vdots              \\
	\vdots   & \vdots     &  \vdots       & \ddots        & \ddots    &  \ddots                    & \vdots               \\
	0        & 0          &  0            & \cdots        & \cdots    &\pi^{s-1}I_{k_{s-1}}& \pi^{s-1}A_{s-1,s}
	\end{pmatrix},$$
		where, the columns are grouped into blocks of sizes $k_0,k_1,...,k_{s-1},m-\sum_{i=0}^{s-1}k_i$. And $\mathcal C$ is of type $1^{k_0}\pi^k-1(\pi^2)^{k_2}...(\pi^{s-1})^{k_{s-1}}.$ It is clear that
	$\vert C\vert=\vert M\vert^{\sum_{i=0}^{s-1}(s-i)k_i}.$
	and the rank of $\mathcal C$ is defined to be 
	$k(\mathcal C)=\sum_{i=0}^{s-1}k_i$.  

The type and rank are the invariants of the code. The  linear code $\mathcal C$ is considered to be free if its rank is equal to the maximum of the ranks of the free submodules of $\mathcal C$. Then the code $\mathcal C$ is considered as a free $R$-submodule which is isomorphic as a module to $R^{k(\mathcal C)}$. If we consider the standard inner product in $R^m$ given by  $x\cdot y =\sum x_i\cdot y_i$ where, $x,y\in R^m$, then, the dual code $\mathcal C^{\perp}$ of $\mathcal C$ is defined by $\mathcal C^{\perp}=\lbrace x\in R^m\mid x\cdot y=0 \textit{ for all } y \in\mathcal  C\rbrace.$ 
 If $\mathcal C\subseteq \mathcal C^{\perp}$ then, the code is self-orthogonal, and if $\mathcal C=\mathcal C^{\perp} $ then, the code is self-dual.

\subsection{Lifted Codes}
Let $\pi$ be a uniformizer of the valuation ring $\mathcal{O}_K$. For each $i\leq n$,

 $$R_i=\mathcal{O}_K/\pi^i\mathcal{O}_K=\lbrace a_0+a_1\pi^1+...+a_{i-1}\pi^{i-1}\mid a_i\in \mathcal{O}_K\rbrace.$$
Since every finite chain ring is isomophic to a nontrivial quotient of rings of integers of $p$-adic fields, then, the $R_i$'s are finite chain rings with maximal ideal $\langle \pi\rangle.$
The \emph{ring of formal power series in $\pi$} with coefficient in a finite chain ring $R$ is defined to be the ring given by: 
$$R[[\pi]]=\left\{a(x)=\sum_{i=0}^{\infty}a_i\pi^i\mid a_i\in R\textit{ for all } i\in\mathbb{N}\right\}.$$
where, the addition and multiplication operator are defined as usual. \emph{The ring of power series} is defined as follows:
$$R_{\infty}=\left\{ \sum_{i=0}^{\infty} a_i\pi^{i} \textit{ ; } a_i \in \mathcal{O}_K\right\}.$$

\begin{theorem}
	The ring of formal power series in $\pi$ with coefficients in a nontrivial quotient rings of integers of $K$ is the ring of integers of $K$
	$$R_{\infty}=\mathcal{O}_K.$$ 
\end{theorem}

\begin{proof}
	Every element $\alpha \in \mathcal{O}_K$ can be written in a unique way as
	$\alpha=\sum_{j=0}^{\infty}b_j\pi^{j}$ where,  $b_j\in \mathbb{F}_{p^r}$, then: 
	$$R_{\infty}=\left\{\sum_{i=0}^{\infty}\left(\sum_{j=0}^{\infty}b_j\pi^j\right)\pi^i\right\}=\left\{\sum_{s=0}^{\infty}\left(\sum_{i+j=s}b_{ij}\right)\pi^s\mid b_{ij}\in\mathbb{F}_{p^r}\right\}.$$
	Consider   $a_{s}=\sum_{i+j=s}b_{ij}$ with $b_{ij}\in \mathbb{F}_{p^r}$, thenm $\sum_{i+j=s}b_{ij}$ with $b_{ij}\in \mathbb{F}_{p^r}$ is a finite sum. Therefore
	$$R_{\infty}=\left\{\sum_{s}^{\infty} a_s\pi^s\mid a_s\in\mathbb{F}_{p^r}\right\} = \mathcal{O}_K.$$ \qed
\end{proof}

	The chain of ideals of $\mathcal{O}_K$ is given as follows (see \cite{iwasawa}):
	$$\lbrace 0\rbrace\subset \cdots  \langle\pi^n\rangle\subset \cdots \subset \langle\pi^2\rangle\subset \langle\pi\rangle\subset 
	\langle \pi^0\rangle=\mathcal{O}_K .$$ 
	Thus $R_\infty$ satisfies the ascending chain condition, i.e; $R_{\infty}=\mathcal{O}_{K}$ is a Noetherian ring. Moreover, it is a Euclidean domain (and therefore a Dedekind domain). Indeed, 
	if $a $ and $b\neq 0$ are in $R_{\infty}$, then there are $q$ and $r$ in $R_{\infty}$ such that $a=bq+r$ and either $r=0$ or $f(r)< f(b)$, where $f$ is a function from $R_{\infty}$ to $\mathbb{Z}^{+}.$
Moreover, $V:R_{\infty}\rightarrow\mathbb{Z}^+$ is the function defined by 
$V(0)=0$ and $V(r)=v(r)$ if $r\neq 0.$\\
If $v(a)\geq v(b)$ then $v(a/b)=v(a)_v(b)\geq 0$ and if  $q=a/b\in R_{\infty}$, then, $r=0$.
Thus, if $v(a)<v(b)$, then, $q=0$ and $r=a$.

A submodule  $\mathcal C$ of  rank $k$ over $R_{\infty}^n$ is called {\itshape $\pi$-adic code of length $n$} and rank $k$. 
	Let $\mathcal C$ be a nonzero linear code over $R_{\infty}$ of length $n$, then any generator matrix of $\mathcal C$ is permutation-equivalent to  a matrix of the following form
	\begin{equation}
	G=\begin{pmatrix}
	\pi^{m_0}I_{k_0} & \pi^{m_0}A_{0,1} & \pi^{m_0}A_{0,2}& \pi^{m_0}A_{0,3}&      & & \pi^{m_0}A_{0,r}\\
	& \pi^{m_1}I_{k_1} & \pi^{m_1}A_{1,2}& \pi^{m_1}A_{1,3}&      & & \pi^{m_1}A_{1,r} \\
	&                  & \pi^{m_2}I_{k_2}& \pi^{m_2}A_{2,3}&      & & \pi^{m_2}A_{2,r}  \\
	&                  &                 & \ddots          &\ddots& &                    \\
	&                  &                 &                 &\ddots&\ddots&                \\
	&                  &                 &                 & & \pi^{m_{r-1}}I_{k_{r-1}}& \pi^{m_{r-1}}A_{r-1,r}
	\end{pmatrix}.
	\label{eq:sd} 
	\end{equation}
	The code $\mathcal C$ with generator matrix of this form is said to be of type 
	$$(\pi^{m_0})^{k_0}(\pi^{m_1})^{k_1}\cdots (\pi^{m_{r-1}})^{k_{r-1}},$$
	where $k=k_0+k_1+\cdots+k_{r-1}$ is called its rank and $k_r=n-k$. 
For two integers $i<j$, a map as in \cite{douglif}: 
\begin{equation}
\begin{array}{cccc}
\Psi^j_i: & R_j & \rightarrow & R_i\\
&\sum_{l=0}^{j-1}a_l\pi^l  & \mapsto & \sum_{l=0}^{i-1} a_l\pi^l.
\end{array}
\label{eq:m} 
\end{equation}
If  $R_j$ is replaced with $R_{\infty}$ then,  $\Psi^{\infty}_i$ is denoted by $\Psi_i$.
For any two elements $a,b\in R_{\infty}$ we have:
$$\Psi_i(a+b)=\Psi_i(a)+\Psi_i(b),\textit{ } \Psi_i(ab)=\Psi_i(a)\Psi_i(b).$$
The two maps $\Psi_i$ and $\Psi^j_i$ can be extended naturally from $R_{\infty}^n$ to $R^n_i$
and $R_j^n$ to $R_i^n$ respectively.

  \begin{remark}
  	Based in the above construction in (\ref{eq:m}) the following series of chain rings is obtained:
  	$$R_{\infty}\rightarrow\cdots\rightarrow R_{s}\cdots\rightarrow R_{s-1}\rightarrow R_{m}\rightarrow\cdots\rightarrow R_1$$
  	for $1\leq t^{'}\leq t \leq e$ ($e$ is the degree of the Eisenstein polynomial $g$) . Note that  
  	 $R_1=\mathcal{O}_K/\pi\mathcal{O}_K \cong \mathbb{F}_{p^r}$
  		, $R_{m}=\mathcal{O}_K/\pi^{m}\mathcal{O}_K$
  		, $R_{s-1}=\mathcal{O}_K/\pi^{s-1}\mathcal{O}_K$
  		and $R_{s}=\mathcal{O}_K/\pi^{s}\mathcal{O}_K$.
  \end{remark}
  
  In this step the lifts of a code $\mathcal C$ over a finite chain ring are defined in a similar way as described in \cite{douglif}, but using this more general setting.
  
  \begin{definition}
  	Let $i,j$ be two integers such that $1\leq i \leq j < \infty.$ An $[n,k]$ code $C_1$ over $R_i$ lifts to an $[n,k]$ code $C_2$ over $R_j$, denoted by $C_1\leq C_2,$ if $C_2$ has a generator matrix $G_2$ where $\Psi_i^j(G_2)$ is a generator matrix of $C_1.$ 
  \end{definition}
It can be proven (the proof in \cite{douglif} can be followed in our more general setting) that $C_1=\Psi_i^j(C_2).$ If $C$ is a $[n,k]$-($\pi$-adic) code, then for any $i<\infty,$  $\Psi_i(C)$ will be called {\it the projection} of $C$.  $\Psi_i(C)$ is denoted by $C^i$ and we have  the following result:
  \begin{lemma}
  	Let $C$ be a linear code over $R_i$ and $\tilde{C}$ be the lifted code of $C$ over $R_j$, where $i<j\leqslant \infty$. Then if $C$ is free over $R_i$ then $\tilde{C}$ is free over $R_j$.  	
  \end{lemma}

\subsection{Cyclic codes}

A cyclic code of length $m$ over the ring of integers $\mathcal{O}_K$ is a linear code $\mathcal C$ such that if $(c_0,c_1,\cdots,c_{m-1})\in \mathcal C$, then $(c_{m-1},c_0,\cdots,c_{m-2})\in \mathcal C$. 
 The codewords of a cyclic code over $\mathcal{O}_K$ are represented as usual by polynomials, more precisely they are the ideals of the ring $\mathcal{O}_K[X]/\langle X^{n}-1\rangle$. Furthermore,  a general construction of lifting cyclic codes which generalize the construction giving in \cite{cyclif} is presented. This general construction allows  to lift cyclic codes over finite fields $\mathbb{F}_{p^r}$ to finite chain rings and to the ring of integers $\mathcal{O}_K$.  We will need the Hensel's lemma (see \cite{Neu} for a proof) for the construction.
 
 \begin{theorem}\label{hensel}{(Hensel's Lemma)}
 	Let $K$ be a finite extension of $\mathbb{Q}_p$ of degree $n$, and let $\mathcal{O}_K$ be the ring of integers of $K$ with maximal ideal $\mathfrak{m}=\langle\pi\rangle$, and residue field $k:=\mathcal{O}_K/\langle\pi\rangle$. Let $f\in\mathcal{O}_{K}[X]$ and let $\overline{f}$ be its image in $k[X]$.
 	Let $\overline{g}$, $\overline{h}$ be two coprime polynomials of $k[X]$ such that $\overline{f}=\overline{g}\overline{h}$, then, there exist $g$, $h\in \mathcal{O}_{K}[X]$ for which $f=g h$ and $g\equiv\overline{g}[\pi]$ and $\overline{h}\equiv h[\pi]$ with $\textit{deg } g=\textit{deg } \overline{g}$.
 \end{theorem}

It is well known that if $\mathcal C$ is a  cyclic code of length $m$ over the finite field $\mathbb{F}_{p^r}=\mathcal{O}_K/\langle\pi\rangle$ then $\mathcal C$ is generated by a monic factor $g(X)$ of $X^n-1$:  
$$X^n-1=\overline{g}(X)\overline{h}(X).$$
Taking into account  the Hensel's Lemma   any decomposition modulo $\pi$ can be generalized to a decomposition modulo $\pi^s$ by 
$X^n-1=g_s(X)h_s(X)[\pi^s]$ 
and therefore to $\mathcal{O}_K$
 as $X^n-1=g(X)h(X).$

\section{Lattices and Codes over Finite Chain Rings}\label{sec:3}

\subsection{Lattices over Integers of $p$-adic Fields}

Let $L$ be a vector space of dimension $n$ over $\mathbb{Q}_p$ and let $\Lambda$ be a $\mathbb{Z}_p$-submodule of $L$ of finite rank associated by a non-degenerate bilinear form 
	$b:\Lambda\times\Lambda\rightarrow \mathbb{Z}_p$. The pair  $(\Lambda,b)$ is called {\it an integral lattice over} $L$.
	The dual lattice of $\Lambda$ over $L$ is given by 
	$$\Lambda^{*}=\lbrace y\in L\textit{ ; } b(y,x) \in \mathbb{Z}_p, \forall x \in \Lambda\rbrace.$$
	$\Lambda$ is unimodular lattice if $\Lambda=\Lambda^{*}.$ 
	If $\Lambda$ is a free lattice with a $\mathbb{Z}_p$-basis $\lbrace x_1,\cdots,x_n\rbrace$, then the matrix given by 
	$G=((x_i,x_j))_{i,j}$
	is the generator matrix corresponding to the lattice $\Lambda.$
	For an integral lattice $\Lambda$, the discriminant group is $d_{\Lambda}=\Lambda^{*}/\Lambda$.
	If $\Lambda$ is free then the discriminant of $\Lambda$ denoted by $disc(\Lambda)$ is
	$$disc(\Lambda)=det(G)=det((x_i,x_j))_{i,j}.$$
	The norm ideal of $\Lambda$ is the $\mathbb{Z}_p$-ideal generated by $\lbrace b(x,x)\mid x\in \Lambda\rbrace$.

Now, let $K$ be a Galois extension over $\mathbb{Q}_p$ of dimension $n$,  $K$ can be seen as a $\mathbb{Q}_p$-vector space of dimension $n$.
Let $\Omega$ be an algebraic closer of $\mathbb{Q}_p$. Since $K$ is a separable extension of $ \mathbb{Q}_p$, there are $n$ distinct $\mathbb{Q}_p$-embeddings $\sigma_1,...,\sigma_n$ from $K$ into $\Omega$. For an element $\alpha\in K$ the norm and the trace maps  are given as follows:
$$\mathrm{N}_{K|\mathbb{Q}_p}(\alpha)=\prod_{i=1}^{n}\sigma_i(\alpha),\qquad
\mathrm{Tr}_{K|\mathbb{Q}_p}(\alpha)=\sum_{i=1}^n \sigma_i(\alpha).$$
Hence, the $\mathbb{Q}_p$-bilinear symmetric form that associates $(x,y)\in K\times K $ to the element $\mathrm{Tr}_{K|\mathbb{Q}_p}(xy)\in \mathbb{Q}_p$ is non-degenerate.

$\mathcal{O}_K$ (the ring of integers of $K$) can be considered as the set of the elements of $K$ which are integral over $\mathbb{Z}_p$, then,  $\mathcal{O}_K$ can be written as follows:
$$\mathcal{O}_K=\mathbb{Z}_p e_1\oplus\cdots \oplus \mathbb{Z}_p e_n$$
where $\{e_1,...,e_n\}$ is a free basis of the $\mathbb{Z}_p$-module $\mathcal{O}_K$. Since for $\alpha\in \mathcal{O}_K$  we can write $\alpha e_i=\sum_{j=1}^n \alpha_{ij}e_j$ where $\alpha_{ij}\in \mathbb{Z}_p$, then,  $\mathrm{Tr}_{K|\mathbb{Q}_p}(\alpha)$  is the trace of the $n\times n $ matrix $\alpha_{ij},$  and $\mathrm{Tr}_{K|\mathbb{Q}_p}(x)\in \mathbb{Z}_p$.

As $\mathcal{O}_K$ is a free $\mathbb{Z}_p$-module of rank $n$, with a basis $\{e_1,...,e_n\}$ over $\mathbb{Z}_p$ then   a generator matrix of the lattice is written as follows:  
$$M=\begin{pmatrix}
\sigma_1(e_1)& \sigma_2(e_1)&\cdots & \sigma_n(e_1)\\
\vdots       &   \vdots     &       & \vdots \\
\sigma_1(e_n)& \sigma_2(e_2)& \cdots & \sigma_n(e_n)
\end{pmatrix}.$$

The discriminant of $K$ over $\mathbb{Q}_p$  is denoted by $D_K$ and it is the discriminant of the lattice $\Lambda_b=(\mathcal{O}_K,b)$. Based on~\cite{book} it fcan be written as follows: 
$$D_K=\mathrm{det}\left(\mathrm{Tr}_{K|\mathbb{Q}_p}\left(e_ie_j\right)_{i,j=1}^{n}\right).$$
If $I$ is an ideal of $\mathcal{O}_K$ then $I$ is a  $\mathbb{Z}_p$-submodule. The following section  considers  the ideals of $\mathcal{O}_K$ as lattices by defining ideal lattices which are the general framework for the construction $A$ of lattices.

\subsection{$\mathbb{Z}_p^n$-Ideal Lattices }

Let $I$ be an ideal of $\mathcal{O}_K$, note that $I$  is also  an $\mathcal{O}_K$-submodule of $K$ different from $\lbrace 0\rbrace$. The norm $\mathrm{N}_{K|\mathbb{Q}_p}(I)$ of $I$ is defined as the $\mathbb{Z}_p$-submodule generated by $\mathrm{N}_{K|\mathbb{Q}_p}(x)$ for all $x\in I$. 
\begin{lemma}
Let $I$ be an ideal of $\mathcal{O}_K$ then $\mathrm{N}_{K|\mathbb{Q}_p}(I)=p^{ri}$ for some $i>0$.
\end{lemma}
\begin{proof}
Since $\mathcal{O}_K$ is a principal ideal domain then every ideal $I$ of $\mathcal{O}_K$
is of the form $I=\langle \pi^i\rangle \textit{ ; } i>0,$ and  $ \mathrm{N}_{K|\mathbb{Q}_p}(\langle \pi\rangle)=p$ then $ \mathrm{N}_{K|\mathbb{Q}_p}(\langle \pi^i\rangle)=(p^r)^{i}=p^{ri}. $
\qed
\end{proof}

\begin{definition}
The lattice $(I,b_{I})$ associated to the ideal $I\subseteq \mathcal{O}_K$ and the symmetric bilinear form $b_{I}:I \times I\rightarrow \mathbb{Z}_p$ is given as follows:
$$b_{I}(x,y)=\mathrm{Tr}_{K|\mathbb{Q}_p}(\alpha xy), \textit{   }\forall x,y\in I$$
where $\alpha$ is an element in  $K$ such that $\sigma_i(\alpha)>0$ $\forall i$, 
is called an ideal lattice.
\end{definition}
A generator matrix of $(I,b_{I})$ is given by 
$$G_I=\begin{pmatrix}
\sqrt{\alpha_1}\sigma_1(u_1) & \sqrt{\alpha_2}\sigma_2(u_1) & \cdots & \sqrt{\alpha_n}\sigma_n(u_1)\\
\vdots &\vdots &\cdots &\vdots \\
\sqrt{\alpha_1}\sigma_1(u_n) & \sqrt{\alpha_2}\sigma_2(u_n) &\cdots & \sqrt{\alpha_n}\sigma_n(u_n)
\end{pmatrix}.$$
Its 
 discriminant  is given by (see \cite{book})
$$\mathrm{disc}(\Lambda_I)= \mathrm{N}_{K|\mathbb{Q}_p}(\alpha)\cdot \mathrm{N}(I)_{K|\mathbb{Q}_p}^2\cdot D_K .$$

\noindent
A lattice $\Lambda$ in $\mathcal{O}_{K}^{n}$ is cyclic if $\Lambda$ is an ideal of $\mathcal{O}_{K}[X]/\langle X^n-1\rangle$.

\subsection{Construction $A$ of Lattices}

Let $R=\mathcal{O}_K/\pi^s\mathcal{O}_K$ be a finite chain ring defined as in Section~\ref{sec:chain} and let $C$ be a code over the ring $R$ of length $m$. We consider the map  $\psi:\mathcal{O}_K\rightarrow R$ the reduction modulo the prime $\pi^s$ given in Section~\ref{sec:codes} such that the preimage of $C$ by $\psi$ is the lifted code of $C$ over $\mathcal{O}_K$. Then  $\psi^{-1}(C)$ is an $\mathcal{O}_K$-module of finite rank (see Section~\ref{sec:codes}) and since $\psi^{-1}(C)$ is a $\mathbb{Z}_p$-submodule, then a lattice can be described as follows:
\begin{definition} Given  a code $C$ over the finite chain ring $R=\mathcal{O}_K/\pi^s\mathcal{O}_K$ and the symmetric bilinear form   
	$b_C=\sum_{i=1}^m \mathrm{Tr}_{K|\mathbb{Q}_p}(\alpha x_iy_i)$ where  $\alpha \in\mathcal{O}_K $ the lattice $\Lambda_C=(\Psi^{-1}(C),b_C)$ is defined as
	as the preimage $\psi^{-1}(C)$  of $C$ in $\mathcal{O}_K^m$   together with the symmetric bilinear form $b_C$.
\end{definition}
\begin{lemma}
	The lattice $\Lambda_C=(\Psi^{-1}(C),b_C)$ is an integral lattice.
\end{lemma}
\begin{proof}
	Let $x,y\in \mathcal{O}_K^m$, then the $Tr_{K/\mathbb{Q}_p}(x_iy_i)\in \mathbb{Z}_p$ for all $i=1,\ldots,m$. Since $\alpha \in\mathcal{O}_K$ then $\mathrm{Tr}(\alpha x_iy_i)$ belongs to $\mathbb{Z}_p$, thus, $b_C(x,y)\in \mathbb{Z}_p$ and therefore, $\Lambda_C$ is an integral lattice.  \qed
\end{proof}
The dual lattice of $(\psi^{-1}(C),b_C)$ is the pair $\Lambda_C^{*}=(\psi^{-1}(C)^{*},b_C)$ defined as follows:
	$$\psi^{-1}(C)^{*}=\lbrace x\in K^m\textit{ ; } b_C(x,y)\in \mathbb{Z}_p, \forall y \in \psi^{-1}(C)\rbrace.$$
Let $A$ and $B$ be two finite $\mathcal{O}_K$-modules such that $B\subset A$ then the quotient $A/B$ is a module of finite length. The invariant of $A/B$  denoted by $\chi(A/B)$(see \cite{book}) is a non-zero ideal of $A$. The following statement is straightforward.
\begin{proposition}
	The discriminant of $\Lambda_{\mathcal C}$ is given by:
	$$\mathrm{disc}(\Lambda_{\mathcal C})=\mathrm{N}(\alpha)^m\cdot  D_K^m\cdot \chi(\mathcal{O}_K^m/{\mathcal C})^2.$$
\end{proposition}

If  $\mathcal C$ is a free code, then the lifted code given as the preimage of $C$ by $\psi$ is a also free, thus $\psi^{-1}(\mathcal C)$ is isomorphic as a module to $ \mathcal{O}_K^{k}$ where $k=k(\psi^{-1}(\mathcal C))$ is the rank of the lifted code of $C$. Then, the following result follows:

\begin{corollary}
	For a free code ${\mathcal C}$ the discriminant of $\Lambda_{\mathcal C}$ is 
	$$\mathrm{disc}(\Lambda_{\mathcal C})=D_K^m((p^r)^{m-k})^2.$$
\end{corollary}
If the Galois extension $K|\mathbb{Q}_p$ is ramified i.e; $n=e$ and $f=1$, and $\pi^n=p$ and 
 ${\mathcal C}_i$ is a self-orthogonal code of length $m$ over a finite chain ring $R_i=\mathcal{O}_K/\pi^s\mathcal{O}_K$. Then  the following result is obtained: 
\begin{lemma}
	The lattice formed by the lifted code of a self-orthogonal code ${\mathcal C}_i$ over $R_i=\mathcal{O}_K/\pi^i\mathcal{O}_K$ is integral with respect to the bilinear form given by
	$$b_{{\mathcal C}_i}=\sum_{i=1}^{m}\mathrm{Tr}_{K|\mathbb{Q}_p}(x_iy_i/p).$$ 
\end{lemma}

\begin{proof}
	Let $x=(x_1,...,x_m)$ and $y=(y_1,...,y_m)$ in $\Lambda_C$, then 
$$
	\rho(x\cdot y)=\rho \left(\sum^m_{i=1}x_iy_i\right)= \sum_{i=1}^m\rho(x_i)\rho(y_i)= \rho(x)\cdot\rho(y)=0 .
$$
	Since $\rho(x)\cdot\rho(y)\in C$ and $C\subset C^{\perp}$, then:
	$$\sum_{i=1}^{m} x_{i}y_i=x\cdot y\equiv 0 \mod  \pi^s. $$
 Since $\pi$ is the only prime above $p$, all conjugates of $\sum_{i=1}^{m}x_i y_i$ must lie in $\pi$, and thus this is also true for its trace. In other words, $\mathrm{Tr}_{K|\mathbb{Q}_p}(x_i y_i)\in \pi^s$, thus
	$\mathrm{Tr}_{K|\mathbb{Q}_p}(x_i y_i) \in p\mathbb{Z}_p.$
	Therefore, by the linearity of the trace we have: 
	$$\langle x,y\rangle=\sum_{i=1}^m \mathrm{Tr}_{K|\mathbb{Q}_p}\left(\frac{x_i y_i}{p}\right) = \dfrac{1}{p} \cdot \mathrm{Tr}_{K|\mathbb{Q}}\left(\sum_{i=1}^m x_i y_i\right)$$
and $\Lambda_C$ is integral.\qed
	\end{proof}
		
		\begin{example} This is an example of the construction above for a $p$-adic cyclotomic field.	
			Let $L$ be the field obtained from $\mathbb{Q}_p$ by adjoining a $p$th root of unity $\zeta$,
			where, $[L:\mathbb{Q}_p]= p-1$.
			The ring of integers of $L$ is given by:
			$$\mathcal{O}_L=\left\lbrace  \alpha=\sum_{i=0}^{p-2} a_i \zeta^i \textit{ ; } \alpha_i\in\mathbb{Z}_p \textit{ for } i=0,1,\cdots,p-2\right\rbrace.$$
			The  maximal ideal of $\mathcal{O}_L$ is $\mathfrak{m}_L=\left\langle 1-\zeta \right\rangle.$
			There exist $(p-1)$ distinct embeddings $\sigma_i : L \rightarrow \mathbb{C}_p$, the trace of an element $\alpha\in L$ over $\mathbb{Q}_p$ is $\mathrm{Tr}_{L|\mathbb{Q}_p}(\alpha)= \sum_{i=1}^{p-1}\sigma_i(\alpha)\in \mathbb{Z}_p$.For $x\in \mathbb{Q}_p(\zeta)$ the element  $\bar{x}$ denotes the complex conjugate of $x$ and we consider the symmetric bilinear form $(x,y) \mapsto\mathrm{Tr}_{L|\mathbb{Q}_p} (x\bar{y} )$.\\
			Let $l$ be the subfield of $L$ such that $l=\mathbb{Q}_p(\zeta+\zeta^{-1})$ and $[L:l]=2$, $[l:\mathbb{Q}_p]=\dfrac{p-1}{2}$.
			Moreover, $\mathrm{Tr}_{L|\mathbb{Q}_p}(x\bar{x})= 2\mathrm{Tr}_{l|\mathbb{Q}_p}(x\bar{x})$. Therefore the bilinear form above is even.
			Let $C$ be a code over a finite chain ring $R\simeq \mathcal{O}_L/(1-\zeta)^s\mathcal{O}_L$. The lattice formed by the preimage of $C$  over $\mathcal{O}_L$ associated with the bilinear form $\mathrm{Tr}_{L|\mathbb{Q}_p}$ is integral since $\mathrm{Tr}_{L|\mathbb{Q}_p}(x)\in\mathbb{Z}_p$ which is also even. Therefore, the lattice is unimodular.
		\end{example}
	
	If we consider now  $\mathcal C$  a cyclic code over a finite chain ring $R$, then  the following result is obtained:
	\begin{theorem}
		Let $\mathcal C$ be a cyclic code over $R$. The lattice $\Lambda_\mathcal C=(\Psi^{-1}(C),b_C)$   is a cyclic lattice of $\mathcal{O}_K$.
	\end{theorem}
	\begin{proof}
		Since $\Psi^{-1}(\mathcal C)$ is a cyclic code of $\mathcal{O}_K$ then the lattice $\Lambda_\mathcal C=(\Psi^{-1}(\mathcal C),b_\mathcal C)$ is cyclic.\qed
	\end{proof}
	
	\begin{example}
		Let $L=\mathbb{Q}_{p^r}$ be the unramified extension of $\mathbb{Q}_p$ of degree $r$ obtained by adjoining to $\mathbb{Q}_p$ a primitive $(p^r-1)$st root of unity. 
		The ring of integers of $L$ is denoted by $\mathcal{O}_L$, the maximal ideal is given by $\mathfrak{m}=\langle p\rangle$ and the residue field is $\mathbb{F}_{p^r}.$
		Let $C$ be a cyclic code over $\mathbb{F}_{p^r}=\mathcal{O}_L/(p),$ then $C$ is generated by a monic factor $g_r(X)$ such that
		$$X^n-1=g_r(X)h_r(X)$$
		of $X^n-1$. Using the Hensel's Lemma  any class of cyclic codes can be generalized from $\mathbb{F}_{p^r}$ to $\mathcal{O}_L$ by 
		$X^n-1=g(X)h(X).$
		Then the lattice formed by the lifted code of $\mathcal C$ is a cyclic lattice over $\mathcal{O}_L$.   
	\end{example}
	
	\section{Conclusions}\label{sec:5}
	This paper presented a general construction of lattices from codes over finite chain rings using $p$-adic fields. The connection between finite chain rings and p-adic fields was highlighted by showing that every finite chain ring of the form $GR(p^n, r)[X]/(g(X), p^{n-1}X^t)$ is isomorphic to a non-trivial quotient of ring integers of p-adic fields.
Based on this connection, the lifting of codes over finite chain rings was generalized. Next, the lattices were defined over $p$-adic integers. Thus, using this definition,  lattices over the ring of integers of a Galois extension of $\mathbb{Q}_p$ from lifted codes over finite chain rings were constructed.

	\bibliography{LamiaCC}{}
	\bibliographystyle{plain}
	
\end{document}